\documentclass{article}
\usepackage[applemac]{inputenc}
\usepackage{cite}
\usepackage{fullpage}
\usepackage{amssymb}
\usepackage{amsmath}
\usepackage{algorithm, algorithmic}
\usepackage{graphicx}
\usepackage{relate}
\newcommand{\occ}{\ensuremath{\mathrm{occ}}}

\newcommand{\bin}{\ensuremath{\mathrm{bin}}}

\newcommand{\fail}{\ensuremath{\mathrm{fail}}}
\newcommand{\num}{\ensuremath{\mathrm{num}}}

\newcommand{\enc}{\ensuremath{\mathrm{enc}}}
\newcommand{\hforward}{\ensuremath{\mathrm{hforward}}}
\newcommand{\hfail}{\ensuremath{\mathrm{hfail}}}
\newcommand{\accept}{\ensuremath{\mathrm{accept}}}

\newcommand{\Next}{\ensuremath{\textsc{Next}}}

\newcommand{\ceil}[1]{\left\lceil{#1}\right\rceil}
\newcommand{\floor}[1]{\left\lfloor{#1}\right\rfloor}

\newtheorem{lemma}{Lemma}
\newtheorem{theorem}{Theorem}
\newcommand{\qed}{\hfill\ensuremath{\Box}\medskip\\\noindent}
\newenvironment{proof}{\noindent\emph{Proof. }}

\title{Fast Searching in Packed Strings\footnote{An externded abstract of this paper appeared at the 20th Annual Symposium on Combinatorial Pattern Matching.}}

\author{Philip Bille\thanks{Supported by the Danish Agency for Science, Technology, and Innovation.}  \\  {\tt phbi@imm.dtu.dk}}

\begin{document}
\maketitle

\begin{abstract}
Given strings $P$ and $Q$ the (exact) string matching problem is to find all positions of substrings in $Q$ matching $P$. The classical Knuth-Morris-Pratt algorithm [SIAM J. Comput., 1977] solves the string matching problem in linear time which is optimal if we can only read one character at the time. However, most strings are stored in a computer in a packed representation with several characters in a single word, giving us the opportunity to read multiple characters simultaneously. In this paper we study the worst-case complexity of string matching on strings given in packed representation. Let $m \leq n$ be the lengths $P$ and $Q$, respectively, and let $\sigma$ denote the size of the alphabet. On a standard unit-cost word-RAM with logarithmic word size we present an algorithm using time 
$$
O\left(\frac{n}{\log_\sigma n} + m + \occ\right).
$$
Here $\occ$ is the number of occurrences of $P$ in $Q$. For $m = o(n)$
this improves the $O(n)$ bound of the Knuth-Morris-Pratt algorithm.
Furthermore, if $m = O(n/\log_\sigma n)$ our algorithm is optimal
since any algorithm must spend at least $\Omega(\frac{(n+m)\log
  \sigma}{\log n} + \occ) = \Omega(\frac{n}{\log_\sigma n} + \occ)$
time to read the input and report all occurrences. The result is
obtained by a novel automaton construction based on the
Knuth-Morris-Pratt algorithm combined with a new compact
representation of subautomata allowing an optimal tabulation-based
simulation. \end{abstract}

\section{Introduction}
Given strings $P$ and $Q$ of length $m$ and $n$, respectively, the
\emph{(exact) string matching problem} is to report all positions of
substrings in $Q$ matching $P$. The string matching problem is perhaps
the most basic problem in combinatorial pattern matching and also one
of the most well-studied, see e.g.,~\cite{KMP1977, BM1977,KR1987,
  BYG1992} for classical textbook algorithms and the surveys
in~\cite{Gusfield1997, NR2002}. The first worst-case $O(n)$ algorithm
(we assume w.l.o.g. that $m \leq n$) is the classical
Knuth-Morris-Pratt algorithm~\cite{KMP1977}. If we assume that we can
read only one character at the time this bound is optimal since we
need $\Omega(n)$ time to read the input. However, most strings are
stored in a computer in a \emph{packed representation} with several
characters in a single word. For instance, DNA-sequences have an
alphabet of size $4$ and are therefore typically stored using $2$ bits
per character with $32$ characters in a $64$-bit word. On packed
strings we can read multiple characters in constant time and hence
potentially do better that the $\Omega(n)$ lower bound for string
matching. In this paper we study the worst-case complexity of packed
string matching and present an algorithm to beat the $\Omega(n)$ lower
bound for almost all combinations of $m$ and $n$.

\subsection{Setup and Results}
We assume a standard unit-cost word RAM with word length $w =
\Theta(\log n)$ and a standard instruction set including arithmetic
operations, bitwise boolean operations, and shifts. The space
complexity is the number of words used by the algorithm, not counting
the input which is assumed to be read-only. All strings in this paper
are over an alphabet $\Sigma$ of size $\sigma$. The \emph{packed
  representation} of a string $A$ is obtained by storing $\Theta(\log
n/\log \sigma)$ characters per word thus representing $A$ in
$O(|A|\log \sigma/\log n) = O(|A|/\log_\sigma n)$ words. If $A$ is
given in the packed representation we simply say that $A$ is a
\emph{packed string}. The \emph{packed string matching problem} is
defined as above except that $P$ and $Q$ are packed strings. In the
worst case any algorithm for packed string matching must examine all
of the words in the packed representation of the input strings. The
algorithm must also report all occurrences of $P$ in $Q$ and therefore
must spend at least $\Omega\left(\frac{n}{\log_\sigma n} +
  \occ\right)$ time, where $\occ$ denotes the number of occurrences of
$P$ in $Q$. In this paper we present an algorithm with the following
complexity.
\begin{theorem}\label{thm:main}
For packed strings $P$ and $Q$ of length $m$ and $n$, respectively,
with characters from an alphabet of size $\sigma$, we can solve the
packed string matching problem in time $O\left(\frac{n}{\log_\sigma n}
  + m + \occ\right)$ and space $O(n^{\varepsilon} + m)$ for any
constant $\varepsilon$, $0 < \varepsilon < 1$.  
\end{theorem}
For $m = o(n)$ this improves the $O(n)$ bound of the
Knuth-Morris-Pratt algorithm. Furthermore, if $m = O(n/\log_\sigma n)$
our algorithm matches the lower bound and is therefore optimal. In
practical situations $m$ is typically much smaller than $n$ and
therefore this condition is almost always satisfied.

\subsection{Techniques}\label{sec:techniques}
The KMP-algorithm~\cite{KMP1977} may be viewed as simulating an
automaton $K$ according to the characters from $Q$ in a left-to-right
order. At each character in $Q$ we use $K$ to maintain the longest
prefix of $P$ matching the current suffix of $Q$. Several improvements
of automaton simulation based algorithms are known, see
e.g~\cite{MP1980, Myers1992,BYG1992, WMM1995}. Typically, these
algorithms partition the automaton into many small subautomata which
can then be quickly simulated fast using either precomputed and
tabulated information and/or the arithmetic and logical operations of
the machine. The idea is then to use the fast simulation for the small
subautomata to obtain a faster simulation of the entire automata. This
approach is often called the ``Four Russian Technique''
~\cite{ADKF1970} when tabulation is used and ``word-level
parallelism'' or ``bitparallelism''~\cite{BaezaYates1989} when only
arithmetic and logical instructions are used. There are (at least) two
central components needed to make this approach effective. First, the
partition into subautomata should allow for efficient distribution of
the computation among the subautomata. Secondly, for tabulation or
word-level parallelism to work efficiently the subautomata must be
encoded compactly.

If we attempt to apply this idea to the KMP-algorithm the above
challenges pose major problems. First, the structure of the
transitions in $K$ does not in general allow us to partition $K$ into
subautomata such that a simulation does not change subautomata too
often. Indeed, for any partition we might be forced to repeatedly
change subautomaton after every group of $O(1)$ characters of $Q$ and
hence end up using $\Omega(n)$ time. Secondly, even if we could design
a suitable partition of $K$ into subautomata, a compact encoding of
the transitions of a subautomata is non-trivial to obtain. An explicit
list of such transitions will not suffice to achieve the bound of
Theorem~\ref{thm:main}. The main contribution of this paper are two
new ideas to overcome these problems.

First, we present the \emph{segment automaton}, $C$, derived from $K$.
In $C$, the states of $K$ are grouped into overlapping intervals of $r
= \Theta(\log n/ \log \sigma)$ states from $K$ such that (almost all
of) the states in $K$ are duplicated in $C$. We show how to
selectively ``copy'' the transitions from $K$ to $C$ such that the
total number of transitions between subautomata never exceeds $O(n/r)$
in the simulation on $Q$. Secondly, we show how to exploit structural
properties of the transitions to represent subautomata optimally. This
allows us to tabulate paths of transitions for all subautomata of size
$<r$ using $O(\sigma^r + m)$ space and preprocessing time. Plugging in
$r = \varepsilon \log n/ \log \sigma$, for constant $\varepsilon > 0$,
this is $O(n^{\varepsilon} + m)$ space and preprocessing time. The
simulation can then be performed in time $O(n/r + \occ) =
O(n/\log_\sigma n + \occ)$ leading to Theorem~\ref{thm:main}.

This main contribution of this paper is theoretical, however, we
believe that both the segment automaton and the compact representation
of automata may prove very useful in practice if combined with ideas
from other algorithms for packed matching.

\subsection{Related Work}
Exploiting packed string representations to speed-up string matching
is not a new idea and is even mentioned in the early papers by Knuth
et al. and Boyer and Moore~\cite{KMP1977, BM1977}. More recently,
several packed string matching algorithms have
appeared~\cite{BaezaYates1989b, TP1997,
  Fredriksson2002,Fredriksson2003, KBN2007,FL2009, FL2009a}. However,
none of these improve the worst-case $O(n)$ bound of the classical
KMP-algorithm.

It is possible to extend the ``super-alphabet'' technique by
Fredriksson~\cite{Fredriksson2002, Fredriksson2003} to obtain a simple
trade-off for packed string matching. The idea is to build an
automaton that, similar to the KMP-automaton, maintains the longest
prefix of $P$ matching the current suffix of $Q$ but allows $Q$ to be
processed in groups of $r$ characters. Each state has $\sigma^r$
outgoing transitions corresponding to all combinations of $r$
characters. This algorithm uses $O(n/r + m\sigma^r)$ time and
$O(m\sigma^r)$ space. Choosing $r = \varepsilon \log_{\sigma} n$ this
is $O(n/\log_\sigma n + mn^{\epsilon})$ time and $O(mn^{\epsilon})$
space. Compared to Theorem~\ref{thm:main} this is a factor $\Theta(m)$
worse in space and only improves the $O(n)$ time bound of the
KMP-algorithm when $m = o(n^{1-\epsilon})$.

Packed string matching is closely related to the area of
\emph{compressed pattern matching} introduced by Amir and
Benson~\cite{AB1992,AB1992a}. Here the goal is to search for an
uncompressed pattern in a compressed text without decompressing it
first. Furthermore, the search should be faster than the naive
approach of decompressing the text first and then using the fastest
algorithm for the uncompressed problem. In \emph{fully compressed
  pattern matching} the pattern is also given in compressed form.
Several algorithms for (fully) compressed string matching are known,
see e.g., the survey by Rytter~\cite{Rytter1999}. For instance, if $Q$
is compressed with the Ziv-Lempel-Welch scheme~\cite{Welch1984} into a
string $Z$ of length $z$, Amir et al.~\cite{ABF1996} showed how to
find all occurrences of $P$ in time $O(m^2 + z)$. The packed
representation of a string may be viewed as the most basic way to
compress a string. Hence, in this perspective we are studying the
fully compressed string matching problem for packed strings. Note that
our result is optimal if the pattern is not packed.

As mentioned above, speeding up automaton based algorithms for string
matching is a well-known idea that has been succesfully applied to a
number of string matching problems. For instance,
Myers~\cite{Myers1992} showed how speed up the simulation of
Thompson's~\cite{Thomp1968} automaton construction for the regular
expression matching problem. The improvement here is from updating a
set of states in a non-deterministic automaton efficiently, whereas we
obtain an improvement by processing multiple characters quickly.  With
few exceptions, see e.g.,~\cite{MP1980, BT2009}, most of the known
improvements of automata-based algorithms are based on updating sets
of states.

\subsection{Outline}
We first briefly review the KMP-algorithm and how it can be viewed as
simulating an automaton in Sec.~\ref{sec:KMP}. We then present the
two major components of our algorithm. Specifically, in
Sec.~\ref{sec:segment} we present the segment automaton and in
Sec.~\ref{sec:representation} we show how to compactly represent
and efficiently tabulate subautomata. In Sec.~\ref{sec:algorithm}
we put the components together and present the complete
algorithm. Finally, in Sec.~\ref{sec:remarks} we conclude with some
remarks and open problems.

\section{The Knuth-Morris-Pratt Automaton and String Matching}\label{sec:KMP}
We briefly review the KMP-algorithm~\cite{KMP1977} which will be the
starting point of our new algorithm. The algorithm that we describe
and use is a slightly simpler version of the KMP-algorithm, often
referred to as the ``Morris-Pratt'' algorithm~\cite{MP1970}, which
suffices to achieve our results.

Let $A$ be a string of length $|A|$ on an alphabet $\Sigma$. The
character at position $i$ in $A$ is denoted $A[i]$ and the substring
from position $i$ to $j$ is denoted by $A[i,j]$. The substrings
$A[1,j]$ and $A[i, |A|]$ are the \emph{prefixes} and \emph{suffixes}
of $A$, respectively.

The Knuth-Morris-Pratt automaton (KMP-automaton), denoted $K(P)$, for
$P$ consists of $m+1$ states identified by the integers $\{0, \ldots,
m\}$ each corresponding to a prefix of $P$. From state $s$ to state
$s+1$, $0 \leq s < m$ there is a \emph{forward transition} labeled
$P[s]$. We call the rightmost forward transition from $m-1$ to $m$ the
\emph{accepting transition}. From state $s$, $0 < s \leq m$, there is
a \emph{failure transition} to a state denoted $\fail(s)$ such that
$P[1,\fail(s)]$ is the longest prefix of $P$ matching a proper suffix
of $P[1,s]$. Fig.~\ref{fig:segment}(a) depicts the KMP-automaton for
the pattern $P = \text{ababca}$. \begin{figure}[t]
  \centering \includegraphics[scale=.5]{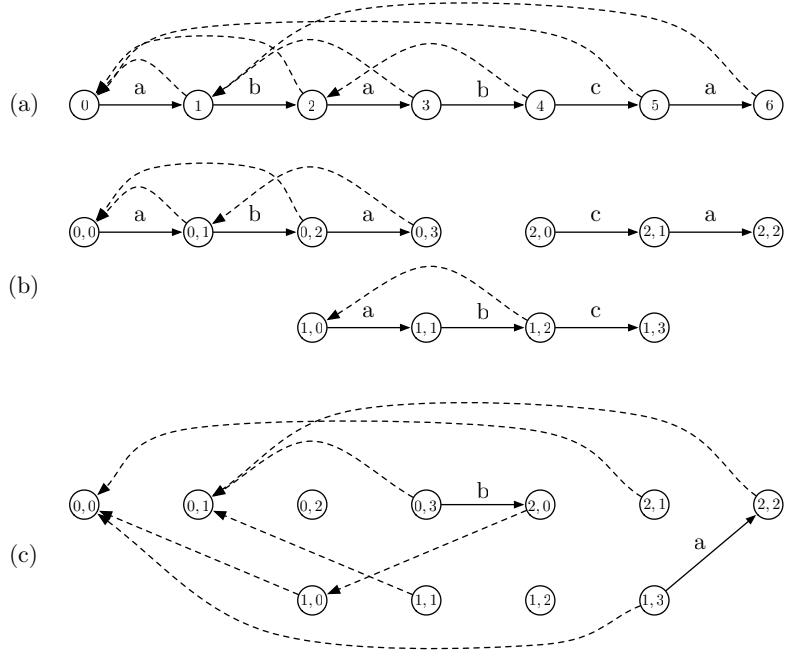}
  \caption{(a) The Knuth-Morris-Pratt automaton $K(P)$ for the pattern
    $P = \text{ababca}$. Solid lines are forward transitions and
    dashed lines are failure transitions. (b)-(c) The corresponding
    segment automaton $C(P, 4)$ for $P$ consisting of $3$ segments
    with $4$, $4$, and $3$ states. The light transitions are shown in
    (b) and the heavy transition transitions in (c).}
   \label{fig:segment}
\end{figure}

The failure transitions form a tree with root in state $0$ and with
the property that $\fail(s) < s$ for any state $s$. Since the longest
prefixes of $P[1, s]$ and $P[1,s+1]$ matching a suffix of $P$ can
increase by at most one character we have the following property of
failure transitions.
\begin{lemma}\label{lem:fail}
  Let $P$ be a string of length $m$ and $K(P)$ be the KMP-automaton
  for $P$. For any state $1 < s < m$, $\fail(s+1) \leq \fail(s) + 1$.
\end{lemma}
We will exploit this property in Sec.~\ref{sec:encoding} to
compactly encode subautomata of the KMP-automaton. The KMP-automaton
can be constructed in time $O(m)$~\cite{KMP1977}.

To find the occurrences of $P$ in $Q$ we read the characters of $Q$
from left-to-right while traversing $K(P)$ to maintain the longest
prefix of $P$ matching a suffix of the current prefix of $Q$ as
follows. Initially, we set the state of $K(P)$ to $0$. Suppose that we
are in state $s$ after reading the $k-1$ characters of $Q$, i.e., the
longest prefix of $P$ matching a suffix of $Q[1,k-1]$ is $P[1,s]$. We
process the next character $\alpha = Q[k]$ as follows. If $\alpha$
matches the label of the forward transition from $s$ the next state is
$s+1$. Furthermore, if this transition is the accepting transition
then $k$ is the endpoint of a substring of $Q$ matching $P$ and we
therefore report an occurrence. Otherwise, ($\alpha$ does not match
the label of the forward transition from $s$ to $s+1$) we recursively
follow failure transitions from $s$ until we find a state $s'$ whose
forward transition is labeled $\alpha$ in which case the next state is
$s'+1$, or if no such state exist we set the next state to be $0$. We
define the \emph{simulation} of $K(P)$ on $Q$ to be sequence of
transitions traversed by the algorithm.

Each time the simulation on $Q$ follows a forward transition we
continue to the next character and hence the total number of forward
transitions is at most $n$. Each failure transition strictly decreases
the current state number while forward transitions increase the state
number by $1$. Since we start in state $0$ the number of failure
transition is therefore at most the number of forward
transitions. Hence, the total number of transitions is at most $2n$
and therefore the searching takes $O(n)$ time. In total the
KMP-algorithm uses time $O(n+m) = O(n)$.

\section{The Segment Automaton}\label{sec:segment}
In this section we introduce a simple automaton called the
\emph{segment automaton}. The segment automaton for $P$ is equivalent
to $K(P)$ in the sense that the simulation on $Q$ at each step
provides the longest prefix $P$ matching a suffix of the current
prefix of $Q$. The segment automaton gives a decomposition of $K(P)$
into subautomata of a given size $r$ such that the simulation on $Q$
passes through at most $O(n/r)$ subautomata. In our packed string
matching algorithm, we will simulate the segment automaton using a fast
algorithm for simulating subautomata of size $r = \Theta(\log_\sigma
n)$, leading to the bound of Theorem~\ref{thm:main}.

Let $K = K(P)$ be the KMP-automaton for $P$. For a even integer
parameter $r$, $1 < r \leq m+1$ we define the segment automaton with
parameter $r$, denoted $C(P, r)$, as follows. Define a \emph{segment}
$S$ to be an interval $S = [l, r]$, $0 \leq l \leq r \leq m$, of
states in $K(P)$ and let $|S| = r - l + 1$ denote the size of
$S$. Divide the $m+1$ states of $K$ into a set of $z =
2\floor{(m+1)/r} + 1$ overlapping segments, denoted $SS = \{S_0,
\ldots, S_{z-1}\}$, where $S_i = [l_i, r_i]$ is defined by
\begin{equation*}
l_i = i \cdot \frac{r}{2} \qquad r_i = \min(l_i + r - 1, m) .
\end{equation*}
Thus, each segment in $SS$ consists of $r$ consecutive states from
$K$, except the last segment, $S_{z-1}$, which may be smaller. Any
state $s$ in $K$ appears in at most $2$ segments and adjacent segments
share $r/2$ states.

The segment automaton $C = C(P,r)$ is obtained by adding $|S|$ states
for each segment $S \in SS$ and then selectively ``copying''
transitions from $K$ to $C$. Specifically, the states of $C$ is the
set of pairs given by
$$
\{(i,j) \mid 0 \leq i < z, 0 \leq j < |S_i|\} .
$$ 
We view each state $(i,j)$ of $C$ as the $j$th state of the $i$th
segment, i.e., state $(i,j)$ corresponds to the state $l_i + j$ in
$K$. Hence, each state in $K$ is represented by $1$ or $2$ states in
$C$ and each state in $C$ uniquely corresponds to a state in $K$.

We copy transitions from $K$ to $C$ in the following way. Let $t = (s,
s')$ be a transition in $K$. For each segment $S_i$ such that $s \in
[l_i, r_i]$ we have the following transitions in $C$:
\begin{itemize}
\item If $s' \in [l_i, r_i]$ there is a \emph{light transition} from $(i, s - l_i)$ to $(i, s' - l_i)$. 
\item If $s' \not\in [l_i, r_i]$ there is a \emph{heavy transition} from $(i, s -l_i)$ to $(i', s' - l_{i'})$, where either $S_{i'}$ is the unique segment containing $s'$ or if two segments contain $s$, then $S_{i'}$ is the segment such that $s' \in [l_{i'}, l_{i'} + r/2]$, i.e., the segment containing $s'$ in the leftmost half.
\end{itemize}
If $t$ is a forward transition with label $\alpha \in \Sigma$ it is
also a forward transition in $C$ with label $\alpha$, if $t$ is a
failure transition it is also a failure transition in $C$, and if $t$
is the accepting transition it is also an accepting transition in $C$.
The segment automaton with $r = 4$ corresponding to the KMP-automaton
of Fig.~\ref{fig:segment}(a) is shown in Fig.~\ref{fig:segment}(b) and
(c) showing the light and heavy transitions, respectively. From the
correspondence between $C$ and $K$ we have that each accepting
transition in a simulation of $C$ on $Q$ corresponds to an occurrence
of $P$ in $Q$. Hence, we can solve string matching by simulating $C$
instead of $K$.

We will use the following key property of the $C$. 
\begin{lemma}\label{lem:heavy}
For a string $P$ of length $m$ and even integer parameter $1 < r \leq m+1$, the simulation of the segment automaton, $C(P,r)$, on a string $Q$ of length $n$ contains at most $O(n/r + \occ)$ heavy and accepting transitions. 
\end{lemma}
\begin{proof}
  Consider the sequence $T$ of transitions in the simulation of $C =
  C(P,r)$ on $Q$. Let $N_{\accept}$ denote the number of accepting
  transitions, and let $N_{\hforward}$ and $N_{\hfail}$ denote the
  number of heavy forward and heavy failure transitions,
  respectively. Each accepting transition in $T$ corresponds to an
  occurrence and therefore $N_{\accept} = \occ$. For a state $(i,j)$
  in $C$ we will refer to $i$ as the \emph{segment number}. Since a
  forward transition in $K$ increases the state number by $1$ in $K$ a
  heavy forward transition increases the segment number by $1$ or $2$
  in $C$. A heavy failure transition strictly decrease the segment
  number. Hence, since we start the simulation in segment $0$, we can
  have at most $2$ heavy failure transitions for each heavy forward
  transition in $T$ and therefore
\begin{equation}\label{eq:Nfail}
N_{\hfail} \leq 2N_{\hforward} . 
\end{equation}

If $N_{\hforward} = 0$ the results trivially follows. Hence, suppose
that $N_{\hforward} > 0$. Before the first heavy forward transition in
$T$ there must be at least $r-1$ light transitions in order to reach
state $(0, r-1)$. Consider the subsequence of transitions $t$ in $T$
between an arbitrary heavy transition $h$ and a forward heavy
transition $f$. The heavy transition $h$ cannot end in segment $z-1$
since there is no heavy forward transition from here. All other heavy
transitions have an endpoint in the leftmost half of a segment and
therefore at least $r-1/2$ light transitions are needed before a heavy
forward transition can occur. Recall that the total number of
transition in $T$ is at most $2n$ and therefore the number of heavy
forward transitions in $T$ is bounded by
\begin{equation}\label{eq:Nforward}
N_{\hforward} \leq 2n/(r-1/2) = 4n/r-1.
\end{equation}
Combining the bound on $N_{\accept}$ with \eqref{eq:Nfail} and
\eqref{eq:Nforward} we have that the total number of heavy and
accepting transitions is
$$
N_{\hforward} + N_{\hfail} + N_{\accept} \leq 3N_{\hforward} + \occ = O(n/r + \occ).
$$\qed
\end{proof}

\section{Representing Segments}\label{sec:representation}
We now describe how to efficiently encode segments how to use the
encoding to efficiently tabulate transitions within segments. In the
following section, we show how to combine this with a simulation of
the segment automaton, leading to the full algorithm for packed
string matching.

\subsection{A Compact Encoding}\label{sec:encoding}
Let $S$ be a segment with $r$ states over an alphabet of size
$\sigma$. We show how to compactly represent all light transitions in
$S$ using $O(r\log \sigma)$ bits. To represent forward transitions we
simply store the labels of the $r-1$ light forward transitions in $S$
using $(r-1)\log \sigma = O(r\log \sigma)$ bits. Next consider the
failure transitions. A straightforward approach is to explicitly store
for each state $s \in S$ a bit indicating if its failure pointer is
light or heavy and, if it is light, a pointer to $\fail(s)$. Each
pointer requires $\ceil{\log r}$ bits and hence the total cost for
this representation is $O(r \log r)$ bits. We show how to improve this
to $O(r)$ bits in the following.

First, we locally enumerate the states in $S$ to $[0, r-1]$. Let $I =
\{i_1, \ldots, i_\ell$\}, $0 \leq i_1 < \cdots < i_\ell < r$, be the
set of states in $S$ with a light failure transition and let $F =
\{f_{i_1}, \ldots, f_{i_\ell}\}$ be the set of failure pointers for
the states in $I$. We encode $I$ as a bit string $B_I$ of length $r$
such that $B_I[j] = 1$ iff $j \in I$. This uses $r$ bits. To represent
$F$ compactly we encode $f_{i_1}$ and the sequence of differences
between consecutive elements $D = d_{i_2}, \ldots, d_{i_{\ell}}$,
where $d_{i_j} = f_{i_j} - f_{i_{j-1}}$. We represent $f_{i_1}$
explicitly using $\ceil{\log r}$ bits. Our representation of $D$
consists of $2$ bit strings. The first string, denoted $B_D$, is the
concatenation of the binary encoding of the numbers in $D$, i.e., $B_D
= \bin(d_{i_2}) \cdots \bin(d_{i_\ell})$, where $\bin(\cdot)$ denotes
standard two's complement binary encoding (the differences may be
negative) and $\cdot$ denotes concatenation. Each number $d_j$ uses at
most $1 + \log |d_j|$ bits and therefore the size of the $B_D$ is at
most
\begin{equation}\label{eq:bd}
|B_D| \leq \sum_{j \in I'} (|\log(d_j)| + 1) < r + \sum_{j \in I'} |\log(d_j)|, 
\end{equation}
where $I' = I \backslash \{i_1\}$. The second bit string, denoted
$B_{D'}$, represents the boundaries of the numbers in $B_D$, i.e.,
$B_{D'}[k] = 1$ iff $k$ is the start of a new number in $B_D$. Thus,
$|B_{D'}| = |B_{D}|$. Note that with $f_1$, $B_D$, and $B_{D'}$ we can
uniquely decode $F$. The total size of the representation is
$\ceil{\log r} + 2|B_D|$ bits.

To bound the size of the representation we show that $|B_D| = O(r)$
implying that the representation uses $\ceil{\log r} + 2\cdot O(r) =
O(r)$ bits as desired. We first bound the sum $\sum_{j \in I'}
|d_j|$. Recall from Lemma~\ref{lem:fail} that the failure function
increases by at most $1$ between consecutive states in $K$. Hence,
over the subsequence $F$ of $<r$ of failure pointers in the range
$[0,r-1]$ the total increase of the failure function can be at most
$r$. Hence, $\sum_{j \in I'} d_j \leq r$. Furthermore, if $f_1 = x$,
for some $x \in [0, r-1]$, the total decrease of $F$ over a segment of
$r$ states is at most $x$ plus the total increase and therefore
$\sum_{j \in I'} d_j \geq -(x+r) \geq -2r$. Hence,
\begin{equation}\label{eq:diff}
\sum_{j \in I'} |d_j| \leq 2r
\end{equation}
Combining \eqref{eq:bd} and \eqref{eq:diff} we have that 
\begin{equation*}
\begin{split}
|B_D| &< r + \sum_{j \in I'} \log |d_j| \leq r +
|I'|\log\left(\frac{\sum_{j \in I'} |d_j|}{|I'|}\right) \\
& < r + r\log(2r/r) = O(r).
\end{split}
\end{equation*}
The second inequality follows from Jensen's inequality combined with
the fact that the logarithm is a concave function. In summary, we have
the following result.
\begin{lemma}\label{lem:encoding}
  All light forward and failure transitions of a segment of size $r$
  can be encoded using $O(r \log \sigma)$ bits.
\end{lemma}


\subsection{Simulating Light Transitions}
We now show how to efficiently simulate multiple light transition within
segments. First, we first introduce a number of important concepts.

Let $C = C(P,r)$ be the segment automaton, and consider the path $p$
of states in the simulation on $C$ from a state $(i,j)$ on some string
$q$. Define the \emph{longest light path} from $(i,j)$ on $q$ to be
the longest prefix of $p$ consisting entirely of light non-accepting
transitions in segment $i$. Furthermore, define the \emph{string
  length} of $p$ to be number of forward transitions in $p$, i.e., the
number of characters of $q$ traversed in $p$. For example, consider
state $(1,1)$ in segment $1$ in Fig.~\ref{fig:segment}. The longest
light path on the string $q = \text{bac}$ is the path $p = (1,1),
(1,2), (1,0), (1,1)$. The transition on c from state $(1,1)$ is a
heavy failure transition and therefore not included in $p$. The string
of $p$ is 2.

Our goal is to quickly compute the string length and endpoint of
longest light paths. Let $S^{\enc}$ be the compact encoding of a
segment $S$ as described above including the label of the forward
heavy transition from the rightmost state in $S$ (if any) and a bit
indicating whether or not the rightmost light transition is accepting
or not. Furthermore, let $j$ be a state in $S$, let $q$ be a string,
and define
\begin{relate}
\item[$\Next(S^{\enc}, j, q)$:] Return the pair $(l, j')$, where $l$
  and $j'$ is the string length and final state, respectively, of the longest
  light path in $S$ from $j$ matching a prefix of $q$.
\end{relate}
For example, if $S^{\enc}$ is the encoding of segment $1$
in Fig.~\ref{fig:segment}, then $\Next(S^{\enc}, 1, \text{bac})$
returns the pair $(2, 1)$.

We will efficiently tabulate $\Next$ for arbitrary strings $q$ of
length $r-1$ (the maximum number of light forward transitions in a
segment of size $r$) as follows. Let $b$ be the total number of bits
needed to represent the input to $\Next$. The string $q$ uses $(r-1)
\ceil{\log \sigma}$ bits and by Lemma~\ref{lem:encoding} $S^\enc$ uses
$O(r \log \sigma + \log \sigma + 1) = O(r \log \sigma)$
bits. Furthermore, the state number $j$ uses $\ceil{\log r}$ bits and
hence $b = O(r\log \sigma + \log r) = O(r \log \sigma)$. 

We tabulate all possible results for $\Next$ using a table $T$
containing $2^b$ entries as follows. Let $Z$ be any input to $\Next$
encoded using the above $b$ bit representation and let $\num(Z)$
denote the non-negative integer in $[0, 2^b-1]$ represented by
$Z$. The table stores at position $T[\num(Z)]$ the result of
$\Next$ on input $Z$. We compute each entry using a standard
simulation in $O(r)$ time and therefore we construct $T$ in $2^b \cdot
O(r) = 2^{O(b)}$ time and space. Hence, if we have $t < 2^w$ space
available for $T$ we may set $r = \frac{1}{c} \cdot \frac{ \log
  t}{\log \sigma}$, where $c > 0$ is an upper bound on the constant
appearing in the $2^{O(b)}$ expression above. Hence, the total space
and preprocessing time now becomes $2^{O(b)} =
2^{\frac{1}{c}\frac{c\log t}{\log \sigma} \log \sigma} = O(t)$.

With $T$ precomputed and stored in memory we can now answer arbitrary
$\Next$ queries for arbitrary encoded segments and strings of length
at most $r-1$ in constant time by table lookup.

\section{The Algorithm}\label{sec:algorithm}
We now put the pieces from the previous sections together to obtain
our main result of Theorem~\ref{thm:main}. Assume that we have $t <
2^w$ space available and choose $r = \Theta(\log t/ \log \sigma)$ as
above for the tabulation. We first preprocess $P$ by computing the
following information:
\begin{itemize}
\item The segment automaton $C(P, r)$ with parameter $r$ and $z =
  2\floor{m+1/r}+1$ segments $SS = \{S_0, \ldots, S_{z-1}\}$.
\item The compact encoding $S^\enc$ for each segment $S \in SS$.
\item The tabulated $\Next$ function for segments with $r$ states and
  input string of length $r-1$.
\end{itemize}
We compute the segment automaton and the compact encodings in $O(m)$
time and space. The tabulation for $\Next$ uses $O(t)$ time and space
and hence the preprocessing uses $O(t + m)$ time and space.

We find the occurrences of $P$ in $Q$ using the algorithm described
below. The main idea is to simulate the segment automaton using the
tabulated $\Next$ function with the segment automaton. At each
iteration of the algorithm we traverse light transitions until we
either have processed $r-1$ characters from $Q$ or encounter a heavy or
accepting transition. We then follow the single next transition and
report an occurrence if the transition is accepting. We repeat this
process until we have read all of $Q$. We note that to implement the
algorithm we need to be able to extract any substring of $\leq r-1$
characters from $Q$ quickly even if the substring does not begin at a
word boundary. We use standard shifts to extract such a substring in
constant time.

\paragraph{\bf Algorithm S}(\emph{Packed String Search}). Let $P$ be a
preprocessed string for a parameter $r$ as above. Given a string $Q$
of length $n$ the algorithm finds all occurrences of $P$ in $Q$.
\begin{description}
\item[S1.] [Initialize] Set $(i,j) \leftarrow (0,0)$ and $k \leftarrow
  1$.
\item[S2.] [Do up to $r-1$ light transitions] Compute $(l, j') \leftarrow
  \Next(S_i^{\enc}, j, Q[k, \min(k+r-1, n)])$. At this point $(i,j')$ is
  the state in the traversal of $C$ on $Q$ after reading the prefix
  $Q[1, k+l]$. All transitions on the string $Q[k, k+l]$ are light and
  non-accepting by the definition of $\Next$.
\item[S3.] [Done?] If $k + l= n$ then terminate. In this case step S2
  exhausted remaining characters from $Q$.
\item[S4.] [Do single transition] Compute $(i^*, j^*)$ by traversing the
  transition $\delta$ from $(i,j')$ on character $Q[k+l+1]$. If
  $\delta$ is a failure transition, set $k^* \leftarrow k + l$ and
  otherwise set $k^* \leftarrow k + l + 1$. If $\delta$ is an
  accepting transition report an occurrence ending at position $k^*$.
\item[S5.] [Repeat] Update $(i,j) \leftarrow (i^*,j^*)$ and $k
  \leftarrow k^*$ and repeat from step S2.
\end{description}
For example, consider running Algorithm S on the segment automaton in
Fig.~\ref{fig:segment} with string $Q = \text{abacacababca}$. Here, $r
= 4$ and we therefore process up to $3$ characters from $Q$ in a
single iteration. Initially, we start at state $(0,0)$ in step S1. In
step S2, we compute the longest path of light and non-accepting
transitions matching $Q[1,3] = \text{aba}$. Since there is such a path
from $(0,0)$ in this segment matching aba, we process all of aba and
end at $(0,3)$. Since more characters remain in $Q$, we continue from
step S3. In step S4, we traverse the light failure transition from
$(0,3)$ to $(0,1)$ since $Q[4] = \text{c}$. In step S5, we update
$(i,j) \leftarrow (0,1)$ and $k \leftarrow 3$ and repeat from step
S2. In step S2, we process all characters from $Q[4,6] = \text{cac}$
ending in $(0,0)$. In step S4, we traverse the light forward
transition from $(0,0)$ to $(0,1)$ on $Q[7] = \text{a}$. In iteration
3, we process characters $Q[8,9] = \text{ba}$ in step S2, ending at state
$(0,3)$. The character $Q[10] = \text{b}$ is not traversed since the
corresponding transition is heavy. Instead, in step S5, we traverse b
to $(2,0)$. In iteration 4, we process the character $Q[11] = \text{c}$ in
step S2, ending at $(2,1)$. We cannot traverse the character $Q[12] =
\text{a}$ since the corresponding transition is accepting. Step S5 traverses
the accepting transition and reports an occurrence. In the final
iteration, no light transitions are traversed in step S2 and the
algorithm terminates in step S3.

It is straightforward to verify that Algorithm S simulates $C(P,r)$ on
$Q$ and reports occurrences whenever we encounter an accepting
transition. In each iteration we either read $r-1$ characters from $Q$
and/or perform a heavy or accepting transition. We can process $r-1$
characters from $Q$ on light transitions at most $\ceil{n/(r-1)}$ and by
Lemma~\ref{lem:heavy} the total number of heavy and accepting
transitions is $O(n/r + \occ)$. Hence, the total number of iterations
is $O(n/r + \occ)$. Since each iteration takes constant time this also
bounds the running time. Adding the preprocessing time and plugging in
$r = \Theta(\log t/ \log \sigma)$ the time becomes
$$
O\left(\frac{n}{r} + t + m + \occ\right) = O\left(\frac{n}{\log_\sigma
    t} + t + m + \occ\right)
$$
with space $O(t + m)$. Hence we have the following result.
\begin{theorem}
  Let $P$ and $Q$ be packed strings of length $m$ and $n$,
  respectively. For a parameter $t < 2^w$ we can solve the packed
  string matching problem in time $O\left(\frac{n}{\log_\sigma t} + t
    + m + \occ\right)$ and space $O(t + m)$.
\end{theorem}
Note that the tabulation is independent of $P$. Hence, if we want to
support multiple searches it suffices to precompute the tables
once. If we plugin $t = n^{\varepsilon}$, for $0 < \varepsilon < 1$,
we obtain an algorithm using time $O\left(\frac{n}{\log_\sigma
    (n^\varepsilon)} + n^\varepsilon + m + \occ\right) =
O\left(\frac{n}{\log_\sigma n} + m + \occ\right)$ and space
$O(n^\varepsilon + m)$ thereby showing Theorem~\ref{thm:main}.

\section{Remarks and Open Problems}\label{sec:remarks}
We have presented an almost optimal solution for the packed string
matching problem on a unit-cost RAM with logarithmic word-length. We
conclude with two challenging open problems.

\begin{itemize}
\item Our algorithm relies on tabulation to compute the $\Next$
  function, and therefore its speed is limited by the amount of space
  we have for tables. Consequently, it cannot take advantage of a
  large word length $w \gg \log n$. We wonder if it is possible to
  obtain a packed string matching algorithm that achieves a speed-up
  over the KMP-algorithm that depends on $w$ rather than $\log n$. In
  particular, it might be possible to come up with an algorithm based
  on \emph{word-level parallelism} (a.k.a. \emph{bit
    parallelism}~\cite{BaezaYates1989}), that uses the arithmetic and
  logical instructions in the word-RAM instead of tables to perform
  the computation.
\item It would be interesting to obtain fast algorithms for related
  packed problems. For instance, we wonder if it is possible to obtain
  a similar speed-up for the multi-string matching
  problem~\cite{AC1975}. The Aho-Corasick algorithm~\cite{AC1975} for
  multi-string matching uses an automaton that generalizes the
  KMP-automaton, however, it appears difficult to extend our
  techniques to this automaton.
\end{itemize}

\bibliographystyle{abbrv}
\bibliography{paper}

\end{document}